\documentclass[11pt,reqno]{amsart}

\usepackage{amscd,amssymb,amsmath,amsthm}
\usepackage[arrow,matrix]{xy}
\usepackage{graphicx}
\usepackage{epstopdf}
\usepackage{color}
\newtheorem{thm}{Theorem}

\newtheorem{lemma}{Lemma}

\newtheorem{rk}{Remark}

\numberwithin{equation}{section} 

\def\s{\sigma}

\def\s{\sigma}

\def\s{\sigma}

\def\a{\alpha}

\def\Z{\mathbb{Z}}

\begin{document}
\title[Tree-hierarchy of DNA and distribution of Holliday junctions]{
Tree-hierarchy of DNA and distribution of Holliday junctions
}

\author{U. A. Rozikov}

\address{U.\ A.\ Rozikov, \\ Institute of mathematics,
29, Do'rmon Yo'li str., 100125, Tashkent, Uzbekistan.}
\email {rozikovu@yandex.ru}

\begin{abstract} We define a DNA as a sequence of $\pm 1$'s and embed it on a path of Cayley tree.
Using group representation of the Cayley tree, we give a hierarchy of a countable
set of DNAs each of which 'lives' on the same Cayley tree. This hierarchy has property that
each vertex of the Cayley tree belongs only to one of DNA. Then we give a model (energy, Hamiltonian) of this set of
DNAs by an analogue of Ising model with three spin values (considered as DNA base pairs) on a set of admissible configurations.
To study thermodynamic properties of the model of DNAs we describe corresponding translation invariant Gibbs measures (TIGM) of the model on the Cayley tree of order two. We show that there is a critical temperature $T_{\rm c}$ such that (i) if temperature $T>T_{\rm c}$ then
 there exists unique TIGM; (ii) if $T=T_{\rm c}$ then there are two TIGMs; (iii) if $T<T_{\rm c}$ then
 there are three TIGMs. Each such measure describes a phase of the set of DNAs. We use these results to study
 distributions of Holliday junctions and branches of DNAs.  In case of very high and very low temperatures we give
 stationary distributions and typical configurations of the Holliday junctions.
\end{abstract}
\maketitle



\section{Introduction}

It is known that (see \cite{book}) genetic information is carried in the linear sequence of nucleotides in DNA.
Each molecule of DNA is a double helix formed from two complementary strands of nucleotides
held together by hydrogen bonds between $G-C$ and $A-T$ base pairs. Duplication of the genetic
information occurs by the use of one DNA strand as a template for formation of a complementary strand.
The genetic information stored in an organism's DNA contains the instructions for all the proteins
the organism will ever synthesize.

In \cite{Be} the author introduced an elastic model of large-scale
duplex DNA structure. The paper \cite{Be1} develops
this approach in detail, deriving explicit expressions for the elastic
equilibrium shapes of stressed DNA and suggesting applications to questions
of supercoiling.

Studying of DNA's thermodynamics one wants to know how temperature affects
the nucleic acid structure of double-stranded DNA \cite{Man}. There are several models of such thermodynamics
of DNAs (see for example, \cite{Ca},  \cite{Pe}, \cite{Ta}).

In \cite{Ta} forty thermodynamic parameters were estimated for DNA
duplexes with a single bulge loop. To investigate the effect
of the type of bulged base and its flanking base pairs, the
nearest-neighbor model to DNA sequences with a single bulge loop is applied.

The Chapter 3 of the book \cite{Pe} deals with the statistics of DNA, a problem
that, recently has attracted the attention of physicists, who attempt to uncover the scaling features of
coding (exons) and noncoding (introns) regions of a gene. Although there is still some
controversy, the current view is that long-range correlations exist when the complete
gene sequence, including introns, is considered and that these correlations are absent
when only the exons are read. The book also contains a chapter on the thermodynamics of DNA where
several results on variants of the Peyrard-Bishop model for local denaturation or melting
of DNA, i.e., the dissociation equilibrium between duplex and single-stranded DNAs,
are derived (see \cite{Pe} and references therein).

In this paper to study thermodynamic properties of a model of DNAs 
we embed them on a Cayley tree.  Motivations of this study will be clear after reading the Section 2,
where I tried to give relations (in footnotes) between biological notations and mathematical ones.

The paper is organized as follows.
In Section 2 we give main definitions from biology (DNA, Holliday junction, branched DNA etc.) and mathematics (Cayley tree, group of its representation, ${\mathbb Z}$-path etc.) which are needed in other sections.
Also we give tree-hierarchy of the set of DNAs.
Then define our model of DNAs, to study its thermodynamics.
In Section 3, we give a system of functional equations, each solution of which
defines a consistent family of finite dimensional Gibbs distributions and guarantees existence of thermodynamic limit for such distributions. Section 4 is devoted to translation invariant Gibbs measures of the set of DNAs on the Cayley tree of order two.
We show that, depending on temperature, number of translation invariant Gibbs measures can be up to three. Note that non-uniqueness of Gibbs measure corresponds to
phase coexistence in the system of DNAs.  In the last section by properties of Markov chains (corresponding to Gibbs measures) we study Holliday junction and branches of DNAs. In case of very high and very low temperatures we give
 stationary distributions and typical configurations of the Holliday junctions.

\section{Definitions and the model}\

{\it About DNA.} Deoxyribonucleic acid (i.e. DNA)\footnote{see https://en.wikipedia.org/wiki/DNA. Some colored pictures of this paper are taken from internet.} is a molecule that carries most of the
genetic instructions used in the growth, development, functioning and reproduction of
all known living organisms and many viruses.  Most DNA molecules consist of two
biopolymer strands coiled around each other to form a double helix.
The two DNA strands are known as
polynucleotides since they are composed of simpler units called nucleotides.
The two strands of DNA in a double helix can be pulled apart like a zipper,
either by a mechanical force or high temperature\footnote{In this paper we
want to give a model of this thermodynamics.}.

 Each nucleotide is composed of a nitrogen-containing nucleobase --either
 cytosine (C), guanine (G), adenine (A), or thymine (T) --as well as a
 sugar called deoxyribose and a phosphate group. The nucleotides are joined to
 one another in a chain by covalent bonds between the sugar of one nucleotide and the phosphate
 of the next, resulting in an alternating sugar-phosphate backbone (see Fig. \ref{f1}).

 According to base pairing rules (A with T, and C with G)\footnote{In our model we denote ``A with T" by ``-1" and ``C with G" by ``+1".},
 hydrogen bonds bind the nitrogenous bases of the two separate polynucleotide strands to make double-stranded DNA. The total amount of related DNA base pairs on Earth is estimated at $5\cdot  10^{37}$, and weighs 50 billion tonnes\footnote{This facts allow us to consider a DNA as a very long sequence consisting of $-1$ and $+1$'s.}.
\begin{figure}
\includegraphics[width=7cm]{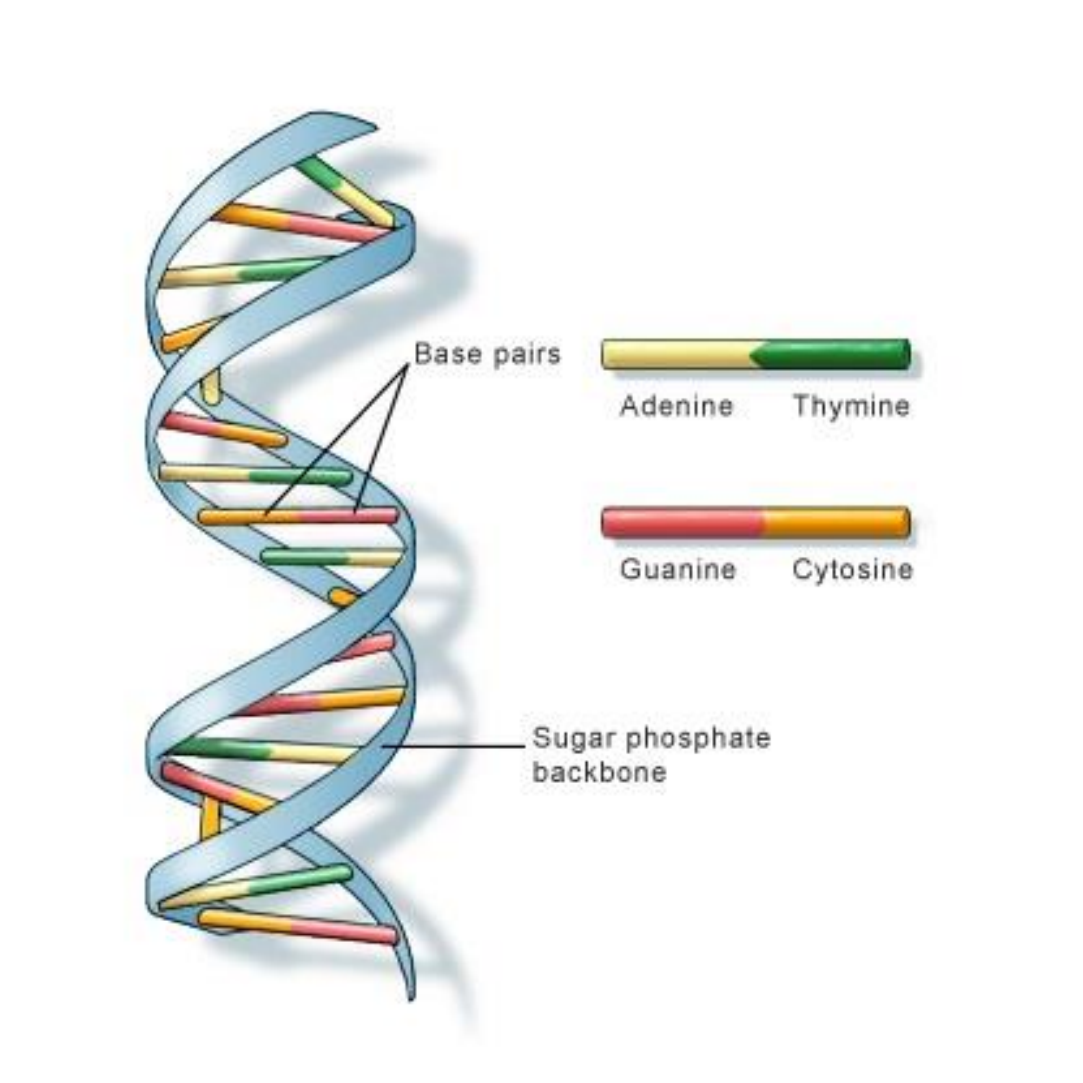},
\caption{\footnotesize \noindent
Structure of a DNA.}\label{f1}
\includegraphics[width=6cm]{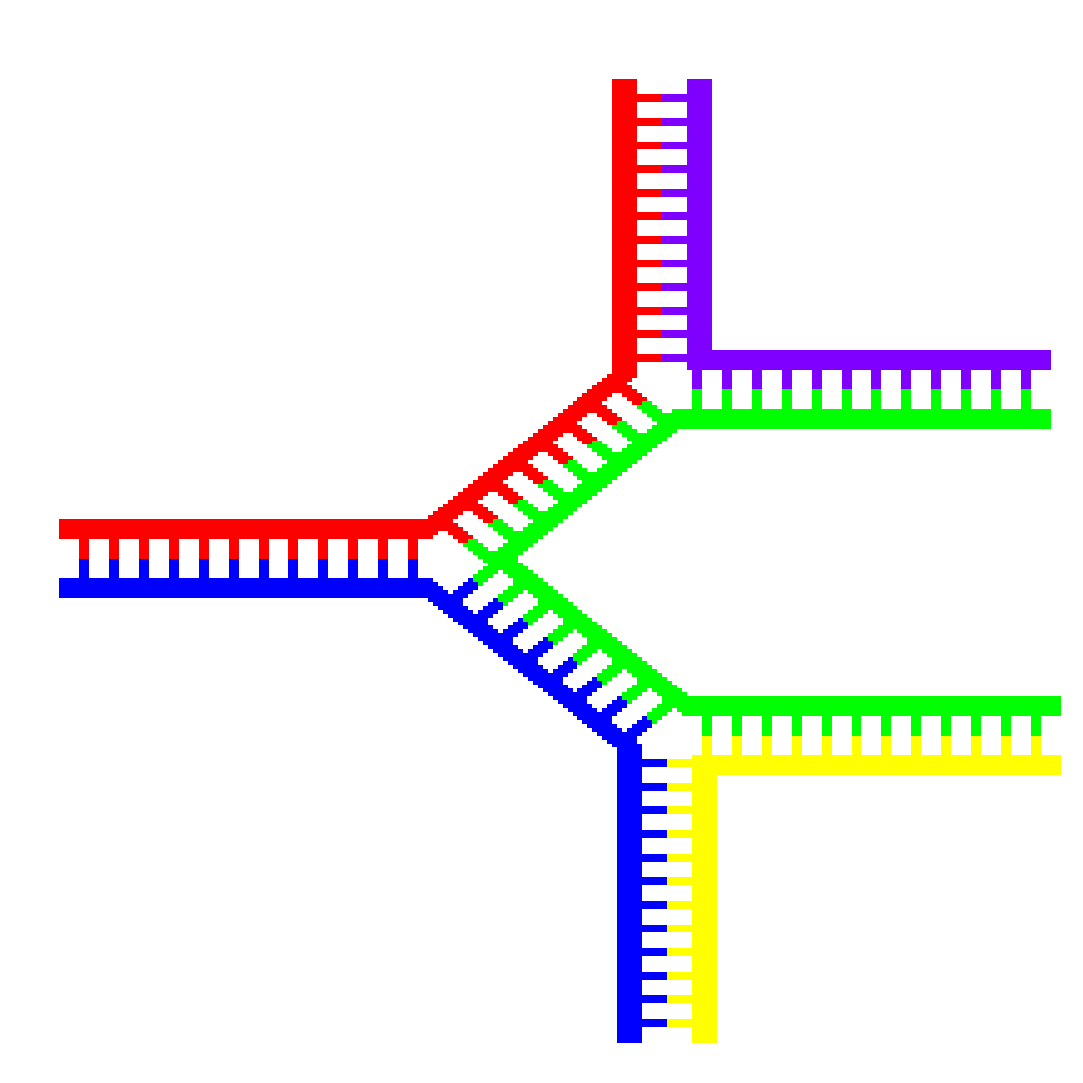}
\caption{\footnotesize \noindent
Multiple branched DNA.}\label{f2}
\includegraphics[width=7cm]{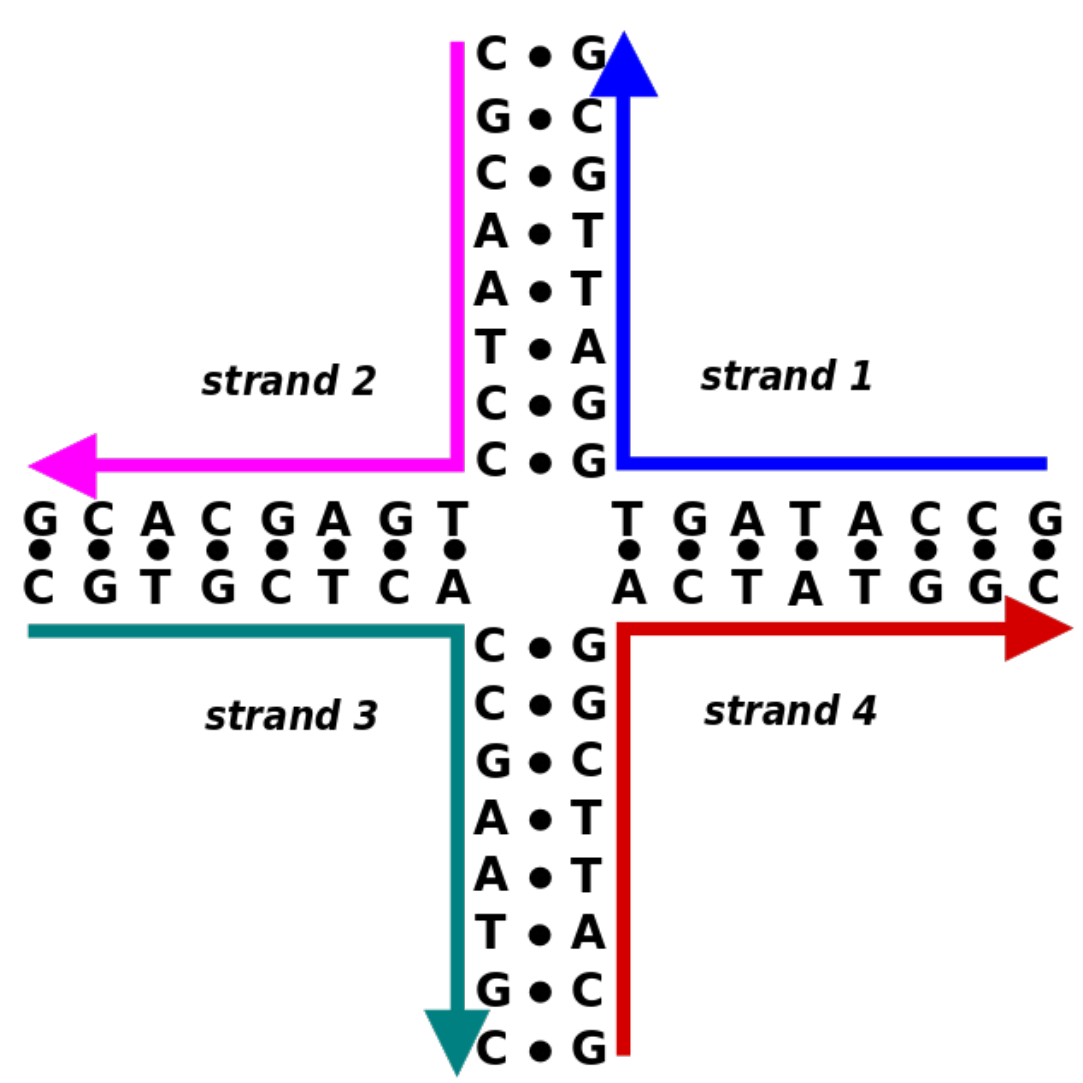},
\caption{\footnotesize \noindent
Schematic of a Holliday Junction showing the base sequence and secondary structure but not the tertiary structure.}\label{f3}
\end{figure}
Branched DNA can occur if a third strand of DNA is introduced and contains adjoining regions able to hybridize with the frayed regions of the pre-existing double-strand. Although the simplest example of branched DNA involves only three strands of DNA, complexes involving additional strands and multiple branches are also possible (see Fig. \ref{f2})  \footnote{Using methods of thermodynamics  (of statistical physics) we describe
a general structure of such branching process.}.

A Holliday junction is a branched nucleic acid structure that contains four double-stranded arms joined together. These arms may adopt one of several conformations depending on buffer salt concentrations and the sequence of nucleobases closest to the junction
(see Fig. \ref{f3}).\footnote{In our model we consider a set of DNAs which 'live' on a tree-like graph. Let $l$ be an edge of this graph we have a function $\sigma(l)$ with three possible values $-1,0,1$ (an analogue of spin values in physical systems), in case $\sigma(l)=0$ we say the edge $l$ does not belong to a DNA. If this $l$ separates two DNA then the value $\sigma(l)=1$ or $\sigma(l)=-1$ means that these two DNA have a Holliday junction.}   In \cite{Ho} the Holliday junction considered as an exquisitely sensitive force sensor whose force response is amplified with an increase in its arm lengths, demonstrating a lever-arm effect at the nanometer-length scale. Mechanical interrogation of the Holliday junction in three different directions helped elucidate the structures of the transient species populated during its conformational changes. In \cite{X} a model is presented for the mechanochemical-coupling mechanism in a Holliday junction.

{\it Cayley tree.} The Cayley tree $\Gamma^k$ of order $ k\geq 1 $ is an infinite tree,
i.e., a graph without cycles, such that exactly $k+1$ edges
originate from each vertex. Let $\Gamma^k=(V,L,i)$, where $V$ is the
set of vertices $\Gamma^k$, $L$ the set of edges and $i$ is the
incidence function setting each edge $l\in L$ into correspondence
with its endpoints $x, y \in V$. If $i (l) = \{ x, y \} $, then
the vertices $x$ and $y$ are called the {\it nearest neighbors},
denoted by $l = \langle x, y \rangle $. The distance $d(x,y), x, y
\in V$ on the Cayley tree is the number of edges of the shortest
path from $x$ to $y$:
$$
d (x, y) = \min\{d \,|\, \exists x=x_0, x_1,\dots, x_{d-1},
x_d=y\in V \ \ \mbox {such that} \ \ \langle x_0,
x_1\rangle,\dots, \langle x_{d-1}, x_d\rangle\} .$$

For a fixed $x^0\in V$ we set $ W_n = \ \{x\in V\ \ | \ \ d (x,
x^0) =n \}, $
\begin{equation}\label{p*}
 V_n = \ \{x\in V\ \ | \ \ d (x, x^0) \leq n \},\ \ L_n = \ \{l =
\langle x, y\rangle \in L \  | \ x, y \in V_n \}.
\end{equation}
For any $x\in V$ denote
$$
W_m(x)=\{y\in V: d(x,y)=m\}, \ \ m\geq 1.
$$
{\it Group representation of the tree.}
Let $G_k$ be a free product of $k + 1$ cyclic groups of the
second order with generators $a_1, a_2,\dots, a_{k+1}$,
respectively, i.e. $a_i^2=e$, where $e$ is the unit element.

It is known that there exists a one-to-one correspondence between the set of vertices $V$ of the
Cayley tree $\Gamma^k$ and the group $G_k$.

 This correspondence can be given as follows. Fix an arbitrary element $x_0\in V$ and correspond it to the unit element $e$ of the group $G_k$. Without loss of generality we assume that the Cayley tree is a planar graph. Using $a_1,\dots,a_{k+1}$ we numerate nearest-neighbors of element $e$, moving by positive direction (see Fig. \ref{fig2}). Now we shall give numeration of the nearest-neighbors of each $a_i$, $i=1,\dots, k+1$ by $a_ia_j$, $j=1,\dots,k+1$. Since all $a_i$ have the common neighbor $e$ we give to it $a_ia_i=a_i^2=e$. Other neighbors are numerated starting from $a_ia_i$ by the positive direction. We numerate the set of all nearest-neighbors of each $a_ia_j$ by words $a_ia_ja_q$, $q=1,\dots,k+1$, starting from $a_ia_ja_j=a_i$ by the positive direction. Iterating this argument one gets
a one-to-one correspondence between the set of vertices $V$ of the
Cayley tree $\Gamma^k$ and the group $G_k$ (see Chapter 1 of \cite{R} for properties of the group $G_k$).
\begin{figure}
  \includegraphics[width=11cm]{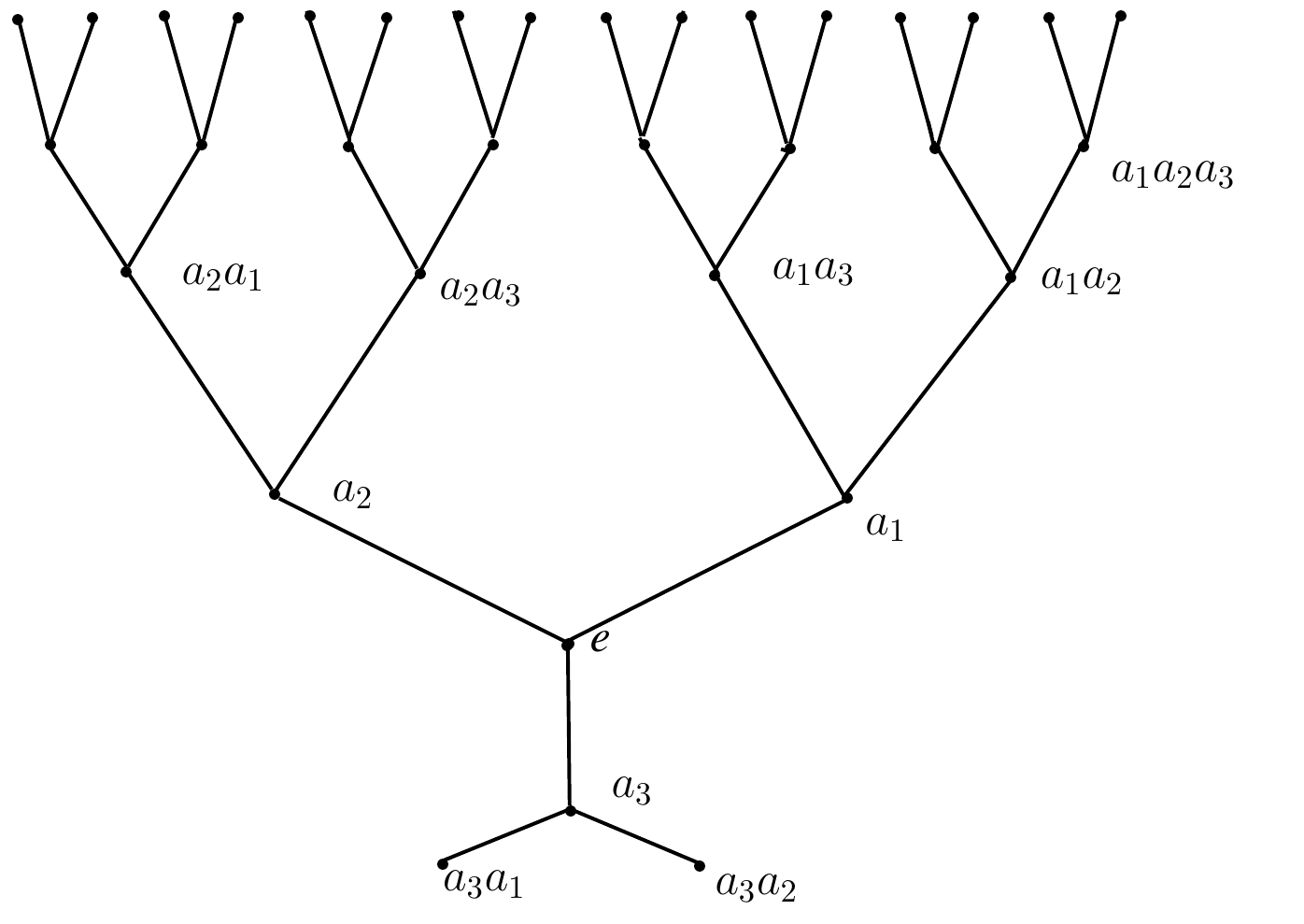}\\
  \caption{Some elements of group $G_2$ on Cayley tree of order two.}\label{fig2}
\end{figure}

We consider a normal subgroup $\mathcal H_0\subset G_k$ of infinite index constructed as follows.
  Let the mapping $\pi_0:\{a_1,...,a_{k+1}\}\longrightarrow \{e, a_1, a_2\}$ be defined by
    $$\pi_0(a_i)=\left\{%
\begin{array}{ll}
    a_i, & \hbox{if} \ \ i=1,2 \\
    e, & \hbox{if} \ \ i\ne 1,2. \\
\end{array}
\right.$$ Denote by $G_1$ the free product of cyclic groups
$\{e,a_1\}, \{e,a_2\}$. Consider
$$f_0(x)=f_0(a_{i_1}a_{i_2}...a_{i_m})=\pi_0(a_{i_1})
\pi_0(a_{i_2})\dots\pi_0(a_{i_m}).$$
    Then it is easy to see that $f_0$ is a homomorphism and hence
    $\mathcal H_0=\{x\in G_k: \ f_0(x)=e\}$ is a normal subgroup of
    infinite index.

Now we consider the factor group
$$G_k/\mathcal H_0=\{\mathcal H_0, \mathcal H_0(a_1), \mathcal H_0(a_2), \mathcal H_0(a_1a_2), \dots\},$$
where $\mathcal H_0(y)=\{x\in G_k: f_0(x)=y\}$. We introduce the notations
$$\mathcal H_n=\mathcal H_0(\underbrace{a_1a_2\dots}_n),$$
$$\mathcal H_{-n}=\mathcal H_0(\underbrace{a_2a_1\dots}_n).$$
In this notation, the factor group can be represented as
$$ G_k/\mathcal H_0=\{\dots, \mathcal H_{-2}, \mathcal H_{-1}, \mathcal H_0, \mathcal H_1, \mathcal H_2, \dots\}.$$
We introduce the following equivalence relation on the set $G_k$: $x\sim y$ if $xy^{-1}\in \mathcal H_0$.
Then $G_k$ can be partitioned to countably many classes $\mathcal H_i$ of equivalent elements.
The partition of the Cayley tree $\Gamma^2$ w.r.t. $\mathcal H_0$ is shown in
Fig. \ref{fig9} (the elements of the class $\mathcal H_i$, $i\in \Z$, are merely denoted by $i$).

{\it $\mathbb Z$-path.}
Denote
$$
q_i(x) = |W_1(x)\cap \mathcal H_i|, \ \ x\in G_k,
$$
where $|\cdot|$ is the counting measure of a set.
We note that (see \cite{RI}) if $x\in \mathcal H_m$, then
$$
q_{m-1}(x)=1,  \ \ q_m(x)=k-1, \ \ q_{m+1}(x)=1.
$$
From this fact it follows that
for any $x\in V$, if $x\in \mathcal H_m$ then there is a unique two-side-path (containing $x$) such that
the sequence of numbers of equivalence classes for vertices of this path
in one side are $m, m+1, m+2,\dots$ in the second side the sequence is $m, m-1,m-2,\dots$.
Thus the two-side-path has the sequence of numbers of equivalent classes as $\mathbb Z=\{...,-2,-1,0,1,2,...\}$.
Such a path is called $\mathbb Z$-path (In Fig. \ref{fig9} one can see the unique $\mathbb Z$-paths of each vertex of the tree.)
In Fig.\ref{tree} the $\mathbb Z$-paths are shown in red color.

Since each vertex $x$ has its own $\mathbb Z$-path one can see that
the Cayley tree considered with respect to normal subgroup $\mathcal H_0$ contains infinitely many (countable) set of
$\mathbb Z$-pathes.
\begin{figure}
   \includegraphics[width=12cm]{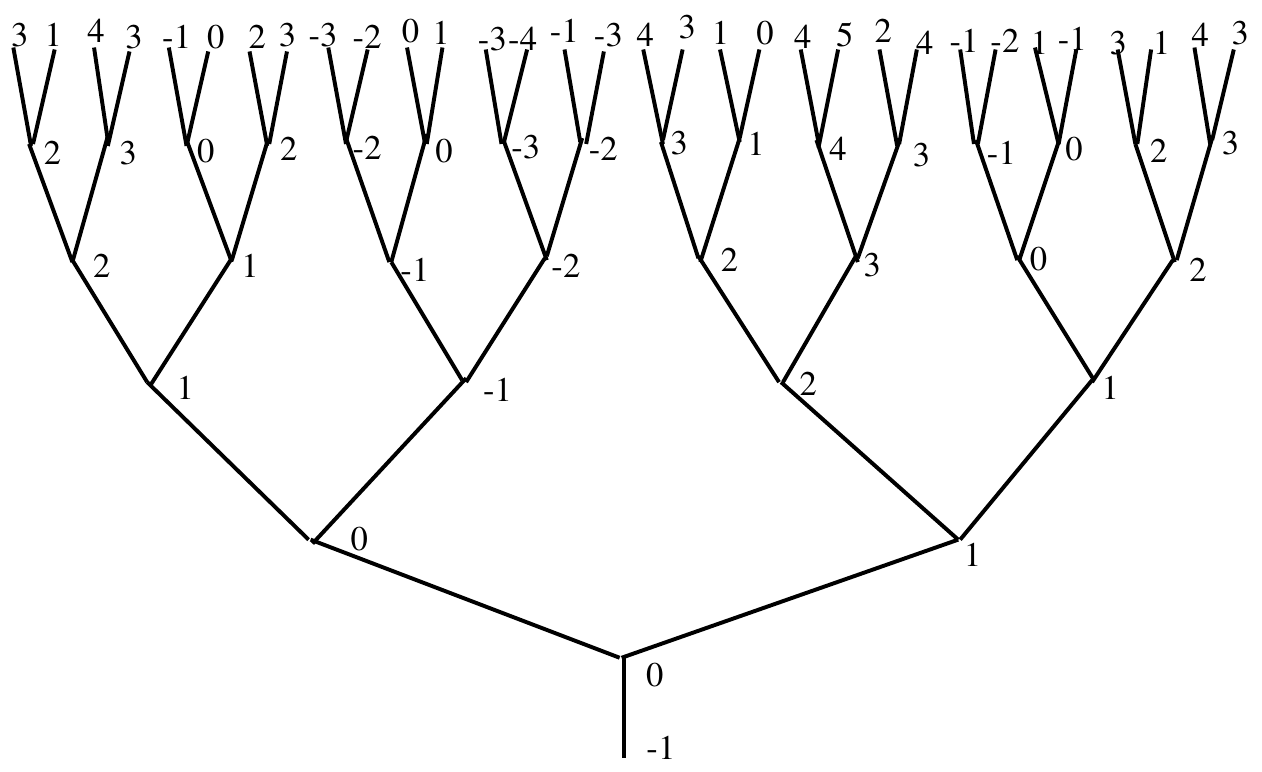}\\
  \caption{The partition of the Cayley tree $\Gamma^2$ w.r.t. $\mathcal H_0$, the elements of the class $\mathcal H_i$, $i\in \Z$, are denoted by $i$.}\label{fig9}
\end{figure}
\begin{figure}
   \includegraphics[width=12cm]{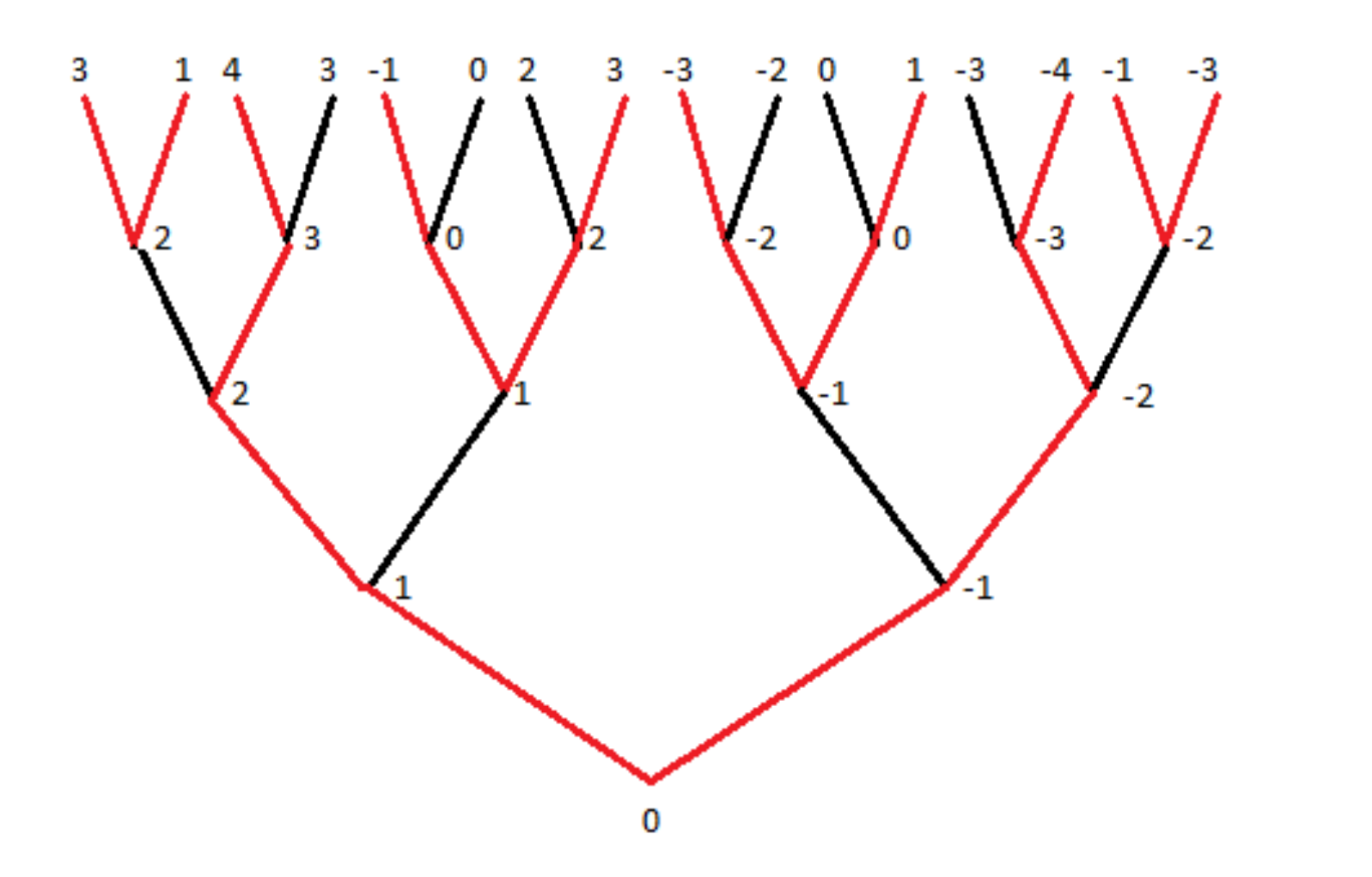}\\
  \caption{Half Cayley tree of order two. The red pathes are $\mathbb Z$-paths, each are two-side
  infinite pathes. The black edges separate $\mathbb Z$-paths. Note that each black edge
  have the same class of endpoints.}\label{tree}
\end{figure}
{\it The model.} Let $L$ be the set of edges of a Cayley tree.
Consider function $\sigma$ which assigns to each edge $l\in L$,
values $\sigma(l)\in \{-1,0,1\}$. Value $\sigma (l)=-1$ (resp. +1)
means that edge $l$ is `occupied' by $-1=$``A with T" (resp. $1=$``C with G"),
and $\sigma (l)=0$ that $l$ is `vacant'.

A configuration $\sigma=\{\sigma(l),\ l\in L\}$ on edges of the Cayley tree
is given by a function from $L$ to $\{-1,0,1\}$. The set of all
configurations in $L$ is denoted by $\Omega$. Configurations in
$L_n$  are defined analogously and the set of all
configurations in $L_n$ is denoted by $\Omega_n$.

A configuration $\sigma=\{\sigma(l),\ l\in L\}$ is called {\it admissible} if
$\sigma(l)\ne 0$ for any $l\in \mathbb Z$-path.\footnote{Note that this admissibility reminds the admissibility in case
of hard-core models (see Chapter 7 of \cite{R} and \cite{SR}), but
they are different, because in the hard-core models the admissibility is given by conditions for configurations on nearest neighboring vertices.}

The restriction of an admissible configuration on
a $\mathbb Z$-path is called {\it a DNA} (since it is a sequence of $-1$ and $1$ values, see above the part "About DNA").

We note that for an admissible configuration $\sigma$ we have $\sigma(l)=0$ iff $l=\langle x,y\rangle$ is with $x\sim y$, i.e.
$f_0(x)=f_0(y)$. Comparing this with Holliday junction one can see that DNAs corresponding to two $\mathbb Z$-paths can
have a junction only through the edge  $l=\langle x,y\rangle$ with $x\sim y$.

The most common discrete models treat DNA as a
collection of rigid subunits representing the base pairs (see Figure \ref{f1}) \cite{Sw}. This
description has long been used by chemists to characterize DNA crystal
structures. Following formulas (4.9) and (4.10) of \cite{Sw} we consider the following
simple model of the energy of the configuration $\sigma$ of a set of DNAs:
\begin{equation}\label{h}
H(\sigma)=J\sum_{\langle l,t\rangle\in L\times L}
\sigma(l)\sigma(t),
\end{equation}
where $J>0$ is a coupling constant, $\sigma(l)\in\{-1,0,1\}$ and
$\langle l,t\rangle$ stands for nearest neighbor edges,
i.e. edges which have a common endpoint.

{\it Tree-hierarchy of the set of DNAs.} We shall give a Cayley tree hierarchy of the set of DNAs, as a DNA crystal.
We note that the Cayley tree is an infinite sphere,
the center of which is everywhere (i.e. each its vertex can be considered as a center),
the circumference nowhere\footnote{There are other
kind of DNA crystals see \cite{Malo}, but the consideration of tree-hierarchy is motivated by the following
quote of Pascal: ``The whole visible world is only an imperceptible atom in the ample
bosom of nature. No idea approaches it. We may enlarge our conceptions beyond all
imaginable space; we only produce atoms in comparison with the reality of things.
It is an infinite sphere, the center of which is everywhere,
the circumference nowhere."  Moreover, ``natural hierarchical" structures in an ultrametric space such that any point
of sphere is his center. This is consistent with Pascal's description the structure of nature. Thus modern science concluding that
structure of the things more close to tree-structure.
The famous legend about under what circumstances Descartes has thought up his system of coordinates:
Having looked at a shady tree through the window protected by lattice rods
the philosopher once, speak, I have suddenly understood what by means of
squares of a lattice can be set numbers of provision of parts of an oak – a trunk,
branches, leaves. And reducing the size of cells of such grid,
it is possible to receive descriptions (or "digitizations" as now speak)
an oak with the increasing and large number of details. The rectangular Cartesian
system of coordinates that is well-known, became discovery of the greatest value
for the subsequent forming of mathematical fundamentals of physics. It is much less known that
if Descartes's thought has gone a little not so and if he, say, has tried to
describe a picture in a window by means of other, new system of numbers –
capable to directly describe a tree thanks to own treelike structure –
that all science could look in a root today differently (see \cite{Pp} and https://kniganews.org/2013/03/25/beyond-clouds-61/)}.

    Given an admissible configuration $\sigma$ on a Cayley
tree, since there are countably many $\mathbb Z$-pathes we have a countable many distinct DNAs.
We say that two DNA {\it are neighbors} if there is an edge (of the Cayley tree) such that
one its endpoint belongs to the first DNA and
another endpoint of the edge belongs to the second DNA. By our
construction it is clear (see Fig. \ref{fig9} and  Fig. \ref{tree}) that
such an edge is unique for each neighboring pair of DNAs. This edge
has equivalent endpoints, i.e. both endpoints belong to the same class
$\mathcal H_m$ for some $m\in \mathbb Z$.

Moreover these countably infinite set of DNAs have a hierarchy that
\begin{itemize}
\item each DNA has its own countably many set of neighboring DNAs.

\item for any two neighboring DNAs, say $D_1$ and $D_2$, there exists a unique
edge $l=l(D_1,D_2)=\langle x,y\rangle$ with $x\sim y$ (of the Cayley tree) which connects DNAs.

\item two DNAs, $D_1,D_2$ (corresponding to two $\mathbb Z$-paths) can
have a Holliday junction if and only if they are neighbors,
 i.e. through the unique edge $l(D_1,D_2)$.

\item  For any finite $n\geq 1$ the ball $V_n$ has intersection only with finitely many DNAs.

\end{itemize}

\section{Finite dimensional distributions and equations}

Let $\Omega_n^a$ (resp. $\Omega^a$) be the set of all
admissible configurations on $L_n$ (resp. $L$).

Denote
$$E_n=\{\langle x,y\rangle\in L: x\in W_{n-1}, \, y\in W_n\},$$
$$\Omega_n^{ba}= \ \ \mbox{the set of admissible configurations on} \ \ E_n.$$
For $l\in E_{n-1}$ denote
$$S(l)=\{t\in E_n: \langle l,t\rangle\}.$$
It is easy to see that
$$S(l)\cap \mathbb Z-{\rm path}=\left\{\begin{array}{lll}
\{l_0,l_1\}\subset L, \ \ \mbox{if} \ \  l\notin \mathbb Z-{\rm path}\\[2mm]
\{l_1\}\subset L, \ \ \ \ \ \ \mbox{if} \ \  l\in \mathbb Z-{\rm path}
\end{array}
\right..$$
We denote
$$S_0(l)=S(l)\setminus \{l_0,l_1\}, \ \ l\notin  \mathbb Z-{\rm path},$$
$$S_1(l)=S(l)\setminus \{l_1\}, \ \ l\in  \mathbb Z-{\rm path}.$$

Define a finite-dimensional distribution of a probability measure $\mu$ on $\Omega_n^a$ as
\begin{equation}\label{*}
\mu_n(\sigma_n)=Z_n^{-1}\exp\left\{\beta H_n(\sigma_n)+\sum_{l\in E_n}h_{\sigma(l),l}\right\},
\end{equation}
where $\beta=1/T$, $T>0$ is temperature,  $Z_n^{-1}$ is the normalizing factor,
$\{h_{i,l}\in \mathbb R, i=-1,0,1, \, l\in L\}$ is a collection of real numbers\footnote{We note that the quantities $\exp(h_{i,l})$ define a boundary
law in the sense of Definition 12.10 in Georgii's book \cite{Ge}. In our case these quantities mean a boundary law of our biological system of DNAs.} and
$$H_n(\sigma_n)=J\sum_{l,t\in L_n:\atop \langle l,t\rangle}
\sigma(l)\sigma(t).$$

We say that the probability distributions (\ref{*}) are compatible if for all
$n\geq 1$ and $\sigma_{n-1}\in \Omega^a_{n-1}$:
\begin{equation}\label{**}
\sum_{\omega_n\in \Omega_n^{ba}}\mu_n(\sigma_{n-1}\vee \omega_n)=\mu_{n-1}(\sigma_{n-1}).
\end{equation}
Here $\sigma_{n-1}\vee \omega_n$ is the concatenation of the configurations.
In this case there exists a unique measure $\mu$ on $\Omega^a$ such that,
for all $n$ and $\sigma_n\in \Omega^a_n$,
$$\mu(\{\sigma|_{L_n}=\sigma_n\})=\mu_n(\sigma_n).$$
Such a measure is called a {\it Gibbs measure} corresponding to the Hamiltonian
(\ref{h}) and functions $h_{i,l}, l\in L$.

The following statement describes conditions on $h_{i,l}$ guaranteeing compatibility of $\mu_n(\sigma_n)$.

\begin{thm}\label{ei} Probability distributions
$\mu_n(\sigma_n)$, $n=1,2,\ldots$, in
(\ref{*}) are compatible iff for any $l\in L$
the following equations hold

 $$z_{0,l}={1+z_{l_0}\over \alpha+\alpha^{-1}z_{l_0}}\cdot {1+z_{l_1}\over \alpha+\alpha^{-1}z_{l_1}}
 \prod_{t\in S_0(l)}{1+z_{0,t}+z_{1,t}\over \alpha+z_{0,t}+\alpha^{-1}z_{1,t}}, \ \   l\notin \mathbb Z-{\rm path},$$
  \begin{equation}\label{***}
  z_{1,l}={\alpha^{-1}+\alpha z_{l_0}\over \alpha+\alpha^{-1}z_{l_0}}\cdot {\alpha^{-1}+\alpha z_{l_1}\over \alpha+\alpha^{-1}z_{l_1}}
 \prod_{t\in S_0(l)}{\alpha^{-1}+z_{0,t}+\alpha z_{1,t}\over \alpha+z_{0,t}+\alpha^{-1}z_{1,t}}, \ \
   l\notin \mathbb Z-{\rm path},
  \end{equation}
  $$
   z_{l}={\alpha^{-1}+\alpha z_{l_1}\over \alpha+\alpha^{-1}z_{l_1}}
 \prod_{t\in S_1(l)}{\alpha^{-1}+z_{0,t}+\alpha z_{1,t}\over \alpha+z_{0,t}+\alpha^{-1}z_{1,t}}, \ \ l\in \mathbb Z-{\rm path}.
$$

Here,
\begin{equation}\label{alp}
\alpha=\exp(J\beta); \, z_{i,l}=\exp\left(h_{i,l}-h_{-1,l}\right), i=0,1; \, z_{l}=\exp\left(h_{1,l}-h_{-1,l}\right).
\end{equation}
\end{thm}
\begin{proof} Fix an edge $l_0\in L$ as a 'root' in $L$. Then there is one-to-one correspondence between elements of $L$ and elements of $V$,
which can be given as follows: let $l=\langle x,y\rangle$ be an arbitrary edge, without loss of generality we assume that $x$ is closer (than $y$)
to $l_0$. Then it is easy to see that the map $l\to y$ is one-to-one. Therefore the proof can be completed similarly as the proof of
Theorem 2.1 in \cite{R}.
\end{proof}

From Theorem \ref{ei} it follows that for any set of vectors
${\bf z}=\{(z_{0,l}, z_{1,l}, z_t), l\notin {\mathbb Z}- {\rm path},  \ \ t\in {\mathbb Z}\}$
satisfying the system of functional equations (\ref{***}) there exists a unique Gibbs measure $\mu$ and vice versa. However,
the analysis of solutions to (\ref{***}) is not easy.

In next section we shall give several solutions to (\ref{***}).

\begin{rk} Note that if there is more than one solution to equation (\ref{***}),
then there is more than one Gibbs measure
corresponding to these solutions. One says that a ``phase'' transition occurs
for the model of DNAs, if equation (\ref{***}) has more than one solution.

The number of the solutions of equation (\ref{***}) depends on  the
temperature-parameter $\alpha=\exp(\frac{J}{T})$. In this paper without loss of generality we take $J=1$.
\end{rk}

\section{Translation invariant Gibbs measures of the set of DNAs}

  In this section, in case $J=1$, we find solutions ${\bf z}_l$ to the system of functional equations (\ref{***}), which does not depend on $l$, i.e.,
  \begin{equation}\label{zzz}
z_{0,l}=u, \ \ z_{1,l}=v, \ \ \forall l\notin {\mathbb Z}-{\rm path}; \ \ z_{l}=w,\ \ \forall l\in {\mathbb Z}-{\rm path}.
\end{equation}
  The Gibbs measure corresponding to such a solution is called {\it translation invariant}.

For $u, v, w$ from (\ref{***}) we get
$$u=\left({1+u+v\over \alpha+u+\alpha^{-1}v}\right)^{k-2}\left({1+w\over \alpha+\alpha^{-1}w}\right)^2,$$
\begin{equation}\label{f}
v=\left({1+\alpha u+\alpha^2v\over \alpha^2+\alpha u+v}\right)^{k-2}\left({1+\alpha^2w\over \alpha^2+w}\right)^2,
\end{equation}
$$w=\left({1+\alpha u+\alpha^2v\over \alpha^2+\alpha u+v}\right)^{k-1}\left({1+\alpha^2w\over \alpha^2+w}\right).$$
Here $u,v,w>0$.
By assumption $J=1$ we have $\alpha>1$.

It is clear that $v=w=1$ satisfies the system (\ref{f}) for any $k\geq 2$ and $\alpha>1$,
then from the first equation of the system we obtain
\begin{equation}\label{vw1}
u=\left({2+u\over \alpha+\alpha^{-1}+u}\right)^{k-2}\left({2\over \alpha+\alpha^{-1}}\right)^2.
\end{equation}

For simplicity we consider the case $k=2$. Then from (\ref{vw1}) we get
$$u=\left({2\over \alpha+\alpha^{-1}}\right)^2=\cosh^{-2}(\beta).$$
Thus for $k=2$ independently on $\alpha$ (i.e. $\beta=1/T$) the vector
$(\cosh^{-2}(\beta),1,1)$ is a solution to system (\ref{f}).
To find other solutions (for $k=2$) we substitute $u$ and $v$ to the third equation
of (\ref{f}) and obtain
$$w={1+\a^2w\over \a^2+w}\left({(\a^2+w)^2 +\a^3(1+w)^2+\a^2(1+\a^2w)^2\over \a^2(\a^2+w)^2 +\a^3(1+w)^2+(1+\a^2w)^2}\right).$$
Solving this equation we get three solutions: $w_1=1$ and
$$w_2=w_2(\a)={\a^6-\a^4-3\a^2-2\a-1-(\a^2-1)\sqrt{(\a^2+1)(\a^3+\a^2+\a+3)(\a^3-\a^2-\a-1)}\over 2(\a^2+\a+1)},$$
$$w_3=w_3(\a)={\a^6-\a^4-3\a^2-2\a-1+(\a^2-1)\sqrt{(\a^2+1)(\a^3+\a^2+\a+3)(\a^3-\a^2-\a-1)}\over 2(\a^2+\a+1)}.$$
Note that solutions $w_2$ and $w_3$ exist iff $\a^3-\a^2-\a-1\geq 0$. It is easy to check that the last inequality is true for any $\a\geq \a_*$ where
$$\a_*={1\over 3}\left(1+(19+3\sqrt{33})^{1/3}+{4\over (19+3\sqrt{33})^{1/3}}\right)\approx 1.839287.$$
Note that both $w_2$ and $w_3$ are positive and $w_2w_3=1$, see Fig.\ref{w23}
for graphs of the solutions as a function of $\alpha\in [\alpha_*,\infty)$.
\begin{figure}
\includegraphics[width=8cm]{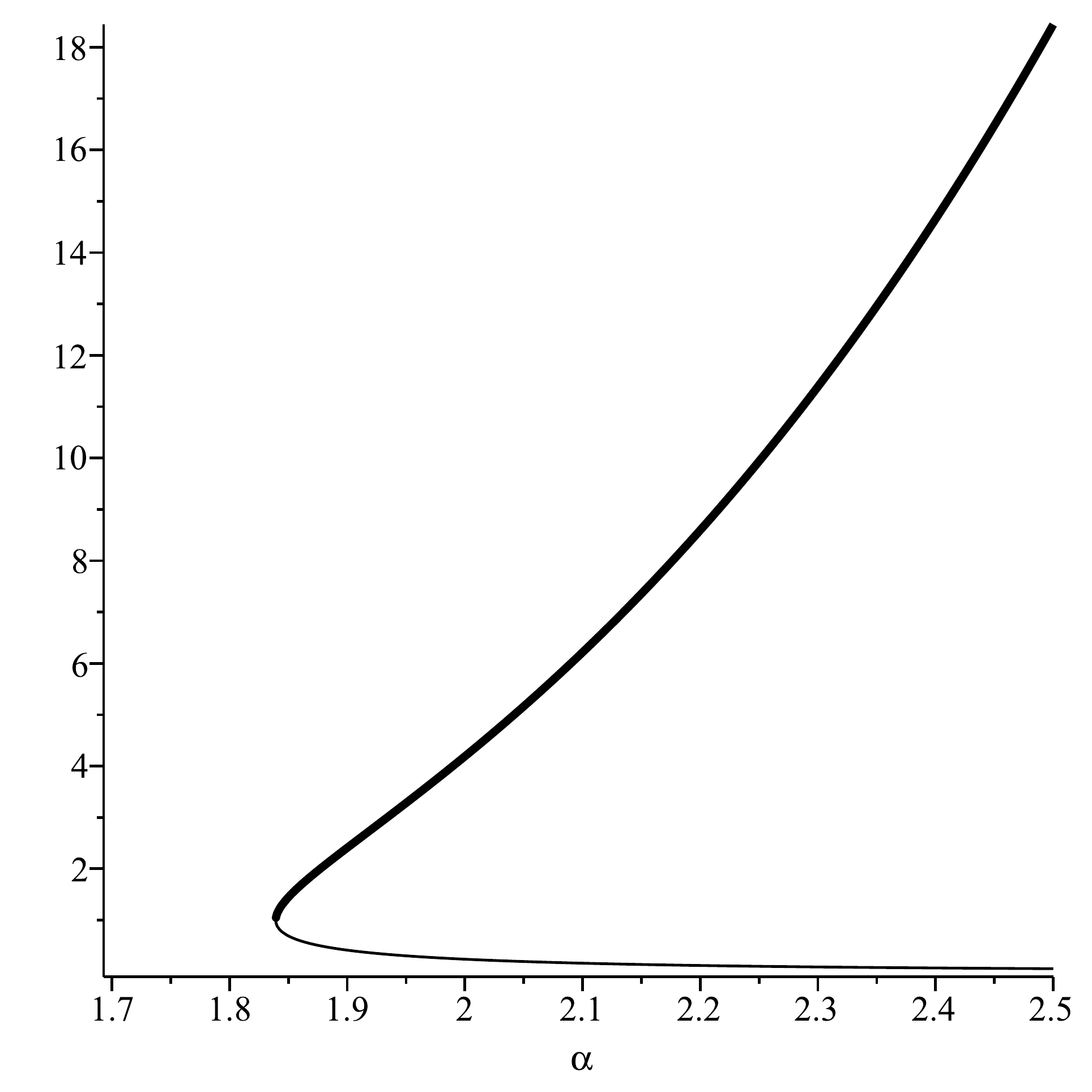}
\caption{The graph of the functions $w_2(\alpha)$ (thin curve) and $w_3(\alpha)$ (bold curve), for $\alpha\in [\alpha_*,\infty)$.
The curves meet at $\alpha_*$ and $w_2(\alpha_*)=w_3(\alpha_*)=1$.}\label{w23}
\end{figure}

Thus, for $k=2$ we proved the following
\begin{lemma}\label{l1} The following hold
\begin{itemize}
\item If $\alpha=\exp(\beta)<\alpha_*$ then system (\ref{f}) has unique solution
$${\bf z_1}=(u_1,v_1,w_1)=(\cosh^{-2}(\beta),1,1);$$
\item If $\alpha=\alpha_*$ then system (\ref{f}) has two solutions
$${\bf z_1}=(\cosh^{-2}(\beta),1,1), \ \  {\bf z_2}=(u_2,v_2,w_2);$$
\item If $\alpha>\alpha_*$ then system (\ref{f}) has three solutions
$${\bf z_1}=(\cosh^{-2}(\beta),1,1), \ \  {\bf z_2}=(u_2,v_2,w_2), \ \ {\bf z_3}=(u_3,v_3,w_3),$$
where
$$u_i=\left({1+w_i\over \alpha+\alpha^{-1}w_i}\right)^2, \ \ v_i=\left({1+\alpha^2w_i\over \alpha^2+w_i}\right)^2, \ \ i=2,3.$$
\end{itemize}
\end{lemma}
Denote by $\mu_i$ the Gibbs measure which, by Theorem \ref{ei}, corresponds to solution ${\bf z_i}$, $i=1,2,3.$

Define critical temperature
$$T_{\rm c}={1\over \ln\alpha_*}.$$

Summarizing above results we obtain the following

\begin{thm}\label{tii} For the model (\ref{h}) of DNAs on the Cayley tree of order $k=2$ the following statements are true
 \begin{itemize}
    \item[(1)] If the temperature $T> T_{\rm c}$ then there is unique translation-invariant Gibbs measure $\mu_1$.

    \item[(2)] If $T= T_{\rm c}$ then there are 2 translation-invariant Gibbs measures $\mu_1$, $\mu_2$.

    \item[(2)] If $T< T_{\rm c}$ then there are 3 translation-invariant Gibbs measures $\mu_1, \mu_2, \mu_3$.
        \end{itemize}
 \end{thm}

 \begin{rk} Analysis of system (\ref{f}) for the case $k\geq 3$ seems difficult. The case $k=2$ is already
 interesting enough to see biological interpretations of Theorem \ref{tii}.
 \end{rk}

 \section{Biological interpretations: Markov chains, Holliday junction and branches of DNA}

For marginals on the two-edge sets which consist of two neighbor edges $l,t$,
considering a boundary law $\{(z_{0,l}, z_{1,l}, z_t), l\notin {\mathbb Z}- {\rm path}, \, t\in {\mathbb Z}\}$ (i.e. the solutions of system (\ref{***})) we have\footnote{This boundary law is normalized at $-1$, i.e., $z_{-1,l}=1$, that is $h_{-1,l}=0$.}
$$\mu(\s(l)=a,\s(t)=b)= \frac{1}{Z} z_{a,l} \exp(\beta ab) z_{b,t}, \ \ a,b=-1,0,1,$$
where $Z$ is normalizing factor.

From this,  the relation between the boundary law and the transition matrix
for the associated tree-indexed Markov chain (Gibbs measure) is immediately obtained
from the formula of the conditional probability. Indeed, once we have ${h}_{a,l}$ (given in formula (\ref{*}))
which is independent on $l$
we define tree-edge-indexed (non-homogenous) Markov chain with states
 $\{-1,0,1\}$  with transition matrix $\mathbb P^{[l,t]}=\left(P^{[l,t]}_{ij}\right)$, depending on $\langle l,t\rangle$  with
 $$P^{[l,t]}_{ij}=\left\{\begin{array}{llll}
 {\exp\left(\beta ij+h_j\right)\over \exp\left(-\beta i+h_{-1}\right)+\exp\left(h_{0}\right)+\exp\left(\beta i+h_{1}\right)}, \ \ \mbox{if} \ \ l,t\notin {\mathbb Z}-{\rm path}, i,j=-1,0,1\\[3mm]
 {\exp\left(\beta ij+h_j\right)\over \exp\left(-\beta i+h_{-1}\right)+\exp\left(\beta i+h_{1}\right)}, \ \ \mbox{if} \ \ l\notin {\mathbb Z}-{\rm path}, \, t\in {\mathbb Z}-{\rm path}, i=-1,0,1; j=-1,1\\[3mm]
 {\exp\left(\beta ij+h_j\right)\over \exp\left(-\beta i+h_{-1}\right)+\exp\left(h_{0}\right)+\exp\left(\beta i+h_{1}\right)}, \ \ \mbox{if} \ \ l \in {\mathbb Z}-{\rm path}, t\notin {\mathbb Z}-{\rm path}, i=-1,1; j=-1,0,1\\[3mm]
 {\exp\left(\beta ij+h_j\right)\over \exp\left(-\beta i+h_{-1}\right)+\exp\left(\beta i+h_{1}\right)}, \ \ \mbox{if} \ \ l,t\in {\mathbb Z}-{\rm path}, i,j=-1,1.
 \end{array}\right.
 $$
 Here $P^{[l,t]}_{ij}$ is the probability to go from a state $i$
at edge $l$ to a state
$j$ at the neighbor edge  $t$.

 Using formulas (\ref{alp}) and (\ref{zzz}) for solutions $(u,v,w)$ to (\ref{f}) we write the matrices $\mathbb P^{[l,t]}=\left(P^{[l,t]}_{ij}\right)$:
 $$\mathbb P^{[l,t]}=\mathbb P_{(3\to 3)}^{[l,t]}=\left(\begin{array}{ccc}
 {\alpha\over \alpha+u+\alpha^{-1}v}& {u\over \alpha+u+\alpha^{-1}v}& {\alpha^{-1} v\over \alpha+u+\alpha^{-1}v}\\[2mm]
 {1\over 1+u+v}& {u\over 1+u+v}& {v\over 1+u+v}\\[2mm]
  {\alpha^{-1}\over \alpha^{-1}+u+\alpha v}& {u\over \alpha^{-1}+u+\alpha v}& {\alpha v\over \alpha^{-1}+u+\alpha v}
\end{array}\right), \ \ \mbox{if} \ \  l,t\notin {\mathbb Z}-{\rm path}.
$$
$$\mathbb P^{[l,t]}=\mathbb P_{(3\to 2)}^{[l,t]}=\left(\begin{array}{ccc}
 {\alpha\over \alpha+\alpha^{-1}w}& 0 & {\alpha^{-1} w\over \alpha+\alpha^{-1}w}\\[2mm]
 {1\over 1+w}& 0& {w\over 1+w}\\[2mm]
  {\alpha^{-1}\over \alpha^{-1}+\alpha w}& 0& {\alpha w\over \alpha^{-1}+\alpha w}
\end{array}\right), \ \ \mbox{if} \ \  l\notin {\mathbb Z}-{\rm path}, t\in {\mathbb Z}-{\rm path}.
$$
 $$\mathbb P^{[l,t]}=\mathbb P_{(2\to 3)}^{[l,t]}=\left(\begin{array}{ccc}
 {\alpha\over \alpha+u+\alpha^{-1}v}& {u\over \alpha+u+\alpha^{-1}v}& {\alpha^{-1} v\over \alpha+u+\alpha^{-1}v}\\[2mm]
 *& *& *\\[2mm]
  {\alpha^{-1}\over \alpha^{-1}+u+\alpha v}& {u\over \alpha^{-1}+u+\alpha v}& {\alpha v\over \alpha^{-1}+u+\alpha v}
\end{array}\right), \ \ \mbox{if} \ \  l\in {\mathbb Z}-{\rm path}, \, t\notin {\mathbb Z}-{\rm path},
$$
where $*$ means that $P_{0j}^{[l,t]}$ is not defined, because $\sigma(l)\ne 0$ for any $l\in \mathbb Z-$path.
$$\mathbb P^{[l,t]}=\mathbb P_{(2\to 2)}^{[l,t]}=\left(\begin{array}{cc}
 {\alpha\over \alpha+\alpha^{-1}w}& {\alpha^{-1} w\over \alpha+\alpha^{-1}w}\\[2mm]
 {\alpha^{-1}\over \alpha^{-1}+\alpha w}&  {\alpha w\over \alpha^{-1}+\alpha w}
\end{array}\right), \ \ \mbox{if} \ \  l\in {\mathbb Z}-{\rm path}, t\in {\mathbb Z}-{\rm path}.
$$
\begin{rk} Matrices $\mathbb P_{(2\to 3)}^{[l,t]}$, $\mathbb P_{(3\to 2)}^{[l,t]}$ do not define a
standard Markov chain.  So we will not consider them. We are interested to see thermodynamics of Holliday junctions and branches of DNAs, for this reason it will be sufficient to study Markov chains generated by
$\mathbb P_{(3\to 3)}^{[l,t]}$ (which gives a Markov chain on the subtree consisting edges which are not on a $\mathbb Z$-path)\footnote{The case $k=2$ is special: in this case  $\mathbb P_{(3\to 3)}^{[l,t]}$ is defined only for $l=t\notin \mathbb Z$-path.}
  and $\mathbb P_{(2\to 2)}^{[l,t]}$ (which gives a Markov chain on the $\mathbb Z$-paths).
 \end{rk}
We note that each matrix $\mathbb P_{(n\to m)}^{[l,t]}$, $n,m=2,3$ is homogenous on the corresponding set of neighbor edges $\langle l,t\rangle$ where it is given, i.e., $\mathbb P_{(n\to m)}^{[l,t]}$ does not depend on $\langle l,t\rangle$ itself but only depends on its relation with ${\mathbb Z}-$path.

It is easy to find the following stationary distributions
$$\pi_{(n\to m)}=(\pi_{(n\to m),-1}, \pi_{(n\to m),0}, \pi_{(n\to m),1})$$ of the matrix $\mathbb P_{(n\to m)}^{[l,t]}$, $n=m$.

$$\pi_{(3\to3)}={1\over N}\left(\begin{array}{ccc}
(\alpha^2+\alpha u+v)(u+\alpha v+\alpha)\alpha\\
(1+u+v)(\alpha^2+\alpha u+v)\alpha u\\
(\alpha u+\alpha^2 v+1)(1+\alpha u+v)v
\end{array}\right)^t,$$
 where $N$ the normalizing factor.

$$\pi_{(2\to 2)}=\left({1+\alpha^{-2}w\over w^2+2\alpha^{-2}w+1}, {w(w+\alpha^{-2})\over w^2+2\alpha^{-2}w+1}\right).$$

The following is known as ergodic theorem for positive
stochastic matrices.
\begin{thm}(\cite{Ge}, p.55) Let $\mathbb P$ be a positive stochastic matrix and $\pi$
the unique probability vector with $\pi \mathbb P=\pi$ (i.e.  $\pi$ is stationary distribution). Then
$$\lim_{n\to \infty} x\mathbb P^n =\pi$$
for all initial vector $x$.
\end{thm}
As corollary of this theorem and above formulas of matrices and stationary distributions we obtain
the following
\begin{thm} In a stationary state\footnote{In the case of non-uniqueness of Gibbs measure
(and corresponding Markov chains) we have different stationary states for
different measures. These depend on the temperature and on the fixed measure.} of the set of DNAs, independently on
$l\notin \mathbb Z$-path, a Holliday junction through $l$ does not occur
with the following probability (with respect to measure $\mu_i$, $i=1,2,3$)
$$\pi_{(3\to3),0}=\pi^{(i)}_{(3\to3),0}={1\over N}
(1+u_i+v_i)(\alpha^2+\alpha u_i+v_i)\alpha u_i.
$$
(Consequently, a Holliday junction occurs with probability $1-\pi^{(i)}_{(3\to3),0}$)\footnote{One can see that $\pi^{(i)}_{(3\to3),0}$ is a function of $i$ and temperature only.}
where $(u_i,v_i)$ are defined in Lemma \ref{l1}.
\end{thm}
\begin{rk} Since each DNA has a countable set of neighbor DNAs, at the same temperature, it my have Holliday junctions
with several of its neighbors. Such connected DNAs can be considered as a
branched DNA. In case of coexistence of more than one Gibbs measures, branches of a DNA can consist
different phases and different stationary states.
\end{rk}

 Now we are interested to calculate
 the limit of  stationary distribution vectors $\pi^{(i)}_{(3\to3)}$, $\pi^{(i)}_{(2\to 2)}$
 (which correspond to the Markov chain generated by the Gibbs measure $\mu_i$)
 in case when temperature $T\to 0$ (i.e. $\beta\to \infty$ and $\alpha\to \infty$) and
when temperature $T\to +\infty$ (i.e. $\beta\to 0$ and $\alpha\to 1$).
To calculate the limit observe that values  $u_{i}, v_i, w_i$, $i=1,2,3$ vary
with $T=1/\beta$.

\begin{lemma}\label{l2} The following equalities hold
\begin{itemize}
\item[-] The case of low temperature:
$$\lim_{T\to 0}\pi_{(3\to 3)}^{(i)}=(1,0,0), \ \ i=1,2;\ \  \lim_{T\to 0}\pi_{(3\to 3)}^{(3)}=(0,0,1).$$
$$\lim_{T\to 0}\pi_{(2\to 2)}^{(1)}=({1\over 2},{1\over 2}), \ \ \lim_{T\to 0}\pi_{(2\to 2)}^{(2)}=(1,0), \ \
\lim_{T\to 0}\pi_{(2\to 2)}^{(3)}=(0,1).$$
\item[-] The case of high temperature\footnote{Recall that measures $\mu_2$ and $\mu_3$ do not exist for $T>T_{\rm c}$}:
$$\lim_{T\to +\infty}\pi_{(3\to 3)}^{(1)}=({1\over 3},{1\over 3},{1\over 3}).$$
$$\lim_{T\to T_{\rm c}}\pi_{(3\to 3)}^{(i)}\approx (0.54246, 0.23555, 0.22199), \, i=1,2,3.$$
$$\lim_{T\to +\infty}\pi_{(2\to 2)}^{(1)}=\lim_{T\to T_{\rm c}}\pi_{(2\to 2)}^{(i)}=({1\over 2},{1\over 2}), i=2,3.$$

\end{itemize}
\end{lemma}
\begin{proof} Since we know explicit formulas for $u_{i}, v_i, w_i$, $i=1,2,3$
the proof consists simple calculations of limits.
\end{proof}

{\bf Structure of DNAs in low and high temperatures.} By Lemma \ref{l2} we have the following structures of the set of DNAs:
\begin{itemize}
\item[(i)] In case $T\to 0$ the set of DNAs have the following stationary states (configurations):
\begin{itemize}
\item[Case $\mu_1$:] All neighboring DNAs connected to each other (Holliday junctions) with state $\sigma(l)=-1$ for any $l\notin \mathbb Z$-path. The sequence of $\pm 1$s, in a DNA on the $\mathbb Z$-path, is free, i.e. can be any sequence, with iid and equiprobable ($=1/2$), of $-1$ and $+1$.
\item[Case $\mu_2$:] All neighboring DNAs connected to each other (Holliday junctions) with state $\sigma(l)=-1$ for any $l\notin \mathbb Z$-path. The sequence of $\pm 1$s, in a DNA on the $\mathbb Z$-path, is rigid, i.e $\sigma(l)=-1$ for all $l\in \mathbb Z$-path. Thus the system contains only one multiple (countable) branched DNA (which has tree structure)\footnote{Recall that in our construction $-1=(A \, with\, T)$, moreover we assumed $(A\, with\, T)=(T\, with\, A)$. Similarly, $+1=(C\, with\, G)=(G \,with\, C)$.}
\item[Case $\mu_3$:] All neighboring DNAs connected to each other with state $\sigma(l)=+1$ for any $l\notin \mathbb Z$-path. The DNA on the $\mathbb Z$-path, is rigid, i.e $\sigma(l)=+1$ for all $l\in \mathbb Z$-path. Thus this is similar to case of $\mu_2$ but all $-1$s replaced by $+1$s.
\end{itemize}
\item[(ii)] In case $T=T_{\rm c}$ the set of DNAs have the following stationary states:
 each neighboring DNAs have Holliday junction with probability $0.76445$ (more precisely, a junction through state $-1$ with probability 0.54246 and a junction through state $+1$ with probability 0.22199) and no junction with probability $0.23555$.   The sequence of $\pm 1$s, in a DNA on the $\mathbb Z$-path, is free, with iid and equiprobable ($=1/2$), of $-1$ and $+1$s.

 \item[(iii)] In case $T\to +\infty$ the set of DNAs have the following stationary states:
 each neighboring DNAs have Holliday junction with probability $2/3$ (more precisely, a junction through state $-1$ or $+1$ with equiprobable 1/3) and no junction with probability $1/3$.  The sequence of $\pm 1$s, in a DNA on the $\mathbb Z$-path, is free, with iid and equiprobable ($=1/2$), of $-1$ and $+1$s.
\end{itemize}

\section*{ Acknowledgements}

The author thanks prof. L.Bogachev and Department of Statistics, School of Mathematics, University of Leeds, UK;
prof. Y.Velenik and the Section of Mathematics, University Geneva, Switzerland;  prof. M. Ladra and the Department of Algebra, 
University of Santiago de Compostela, Spain for financial support and kind hospitality during my visits to these universities.


\begin{thebibliography}{99}

\bibitem{book} B. Alberts, A. Johnson, J. Lewis, M. Raff, K. Roberts, and P. Walter, \textit{Molecular Biology of the Cell}. 4th edition.
New York: Garland Science; 2002.

\bibitem{Be} C.J. Benham, \textit{Elastic model of supercoiling}, Proc. Natl. Acad. Sci. U.S.A., \textbf{74}
(1977),  2397--2401.

\bibitem{Be1} C.J. Benham, \textit{An elastic model of the large-scale structure of duplex DNA}, Biopolymers,
\textbf{18} (1979),  609-–623.

\bibitem{Ca} E. Carlon, \textit{Thermodynamics of DNA microarrays}. Stochastic models in biological sciences, 229--233, Banach Center Publ., 80, Polish Acad. Sci., Warsaw, 2008.

\bibitem{Ge} H.O. Georgii, \textit{Gibbs Measures and Phase Transitions},  Second edition. de Gruyter Studies in Mathematics, 9. Walter de Gruyter, Berlin, 2011.

\bibitem{Ho} S. Hohng, R.Zhou, M.K. Nahas, J. Yu, K. Schulten, D.M.J. Lilley, T.Ha,  \textit{Fluorescence-force spectroscopy maps two-dimensional reaction landscape of the holliday junction.} Science. \textbf{318} (2007), 279--283.

\bibitem{Pp} A.Yu.  Khrennikov, \textit{Non-Archimedean analysis: quantum paradoxes, dynamical systems and biological models}. Dordreht: Kluwer Acad. Publ., 1997.

\bibitem{Man} M. Mandel, J. Marmur, \textit{Use of Ultraviolet Absorbance-Temperature Profile for Determining the Guanine plus Cytosine Content of DNA}. Methods in Enzymology. \text{12} (2) (1968), 198-–206.

\bibitem{Malo} J. Malo, J.C. Mitchell, C. V\'enien-Bryan, J.R. Harris, H. Wille, D.J. Sherratt, A.J. Turberfield,
\textit{Engineering a 2D protein-DNA crystal}. Angew. Chem. Int. Ed. \textbf{44} (2005) 3057--61.


\bibitem{Pe} J.K. Percus, \textit{Mathematics of genome analysis}. Cambridge Studies in Mathematical Biology, 17. Cambridge University Press, Cambridge, 2002.

\bibitem{RI} U.A. Rozikov, F.T. Ishankulov, \textit{Description of periodic $p$-harmonic functions on Cayley trees},
Nonlinear Diff. Equations Appl. \textbf{17}(2)  (2010), 153--160.

\bibitem{R} U.A. Rozikov, Gibbs measures on Cayley trees. {\sl World Sci. Publ}. Singapore. 2013, 404 pp.

\bibitem{SR}  Yu.M. Suhov, U.A. Rozikov, \textit{A hard-core model on a Cayley tree:
an example of a loss network},  Queueing Syst. \textbf{46}(1/2) (2004), 197--212.

\bibitem{Sw} D. Swigon, \textit{The Mathematics of DNA Structure, Mechanics, and Dynamics}, IMA Volumes in Mathematics and Its Applications, \textbf{150} (2009) 293--320.

 \bibitem{Ta}  F. Tanaka, A. Kameda, M. Yamamoto, A. Ohuchi, \textit{Nearest-neighbor thermodynamics of DNA sequences with single bulge loop.} DNA computing, 170--179, Lecture Notes in Comput. Sci., 2943, Springer, Berlin, 2004.

\bibitem{X} P. Xie, \textit{Model for RuvAB-mediated branch migration of Holliday junctions.} J.Theor. Biology. \textbf{249} (2007) 566-573.

\end{thebibliography}
\end{document}